\documentclass[submission,copyright,creativecommons]{eptcs}
\providecommand{\event}{NCMA 2022} 
\usepackage{verbatim}
\usepackage{iftex}

\ifpdf
  \usepackage{underscore}         
  \usepackage[T1]{fontenc}        
\else
  \usepackage{breakurl}           
\fi

\usepackage{amsfonts}
\usepackage{amssymb}
\usepackage{amsmath}
\usepackage[dvipsnames]{xcolor}
\usepackage{graphicx}
\usepackage{tikz}
\usetikzlibrary{positioning,arrows.meta,calc,fit}
\usepackage{enumitem}

\usepackage{amsthm}
\newtheorem{remark}{Remark}
\newtheorem{theorem}{Theorem}
\newtheorem{lemma}{Lemma}
\newtheorem{example}{Example}
\newtheorem{definition}{Definition}

\title{{P versus B:} \\  P Systems as a Formal Framework\\
  for Controllability of Boolean Networks%
}

\def\titlerunning{P Systems as a Formal Framework for 
Controllability of Boolean Networks}

\author{Artiom Alhazov
  \institute{%
    Vladimir Andrunachievici Institute\\of Mathematics and Computer Science\\
    Academiei 5, Chi\c sin\u au, MD-2028, Moldova%
  }
  \email{artiom@math.md}
  \and
  Rudolf Freund
  \institute{%
    Faculty of Informatics, TU Wien\\
    Favoritenstra\ss{}e 9--11, 1040 Wien, Austria%
  }
  \email{rudi@emcc.at}
  \and
  Sergiu Ivanov
  \institute{
    Universit\'{e} Paris-Saclay, Univ. \'{E}vry, IBISC\\
    23, boulevard de France 91034 \'{E}vry, France%
  }
  \email{sergiu.ivanov@ibisc.univ-evry.fr}
}
\def\authorrunning{A. Alhazov, R. Freund, S. Ivanov}

\newcommand{\dottimes}{\mathop{\dot{\times}}}

\begin{document}
\maketitle

\begin{abstract}
  Membrane computing and P systems are a paradigm of massively
  parallel natural computing introduced by Gheorghe P\u{a}un in 1999,
  inspired by the structure of the living cell and by its biochemical
  reactions.  In spite of this explicit biological motivation, P systems 
  have not been extensively used in modelling real-world systems.  
  To confirm this intuition, we establish a state of the art investigation 
  comparing the use of P systems to that of Boolean networks in this
  line of research.  We then propose to use P systems as a tool for
  setting up formal frameworks to reason about other formalisms, and
  we introduce Boolean P systems, specifically tailored for capturing
  sequential controllability of Boolean networks.  We show how to
  tackle some technical challenges and prove that sequential
  controllability properly embeds in the framework of Boolean
  P systems.
\end{abstract}



\section{Introduction}
\label{sec:intro}

Membrane computing and P systems are a paradigm of massively 
parallel computing introduced more than two decades ago by
Gheorghe P\u{a}un~\cite{Paun2000}, and inspired by the structure and
function of the biological cell.  Following the example of the cell,
a membrane (P) system is a hierarchical membrane structure defining
compartments containing multisets of objects, representing the 
biochemical species  in an abstract sense. Multiset rewriting rules are
attached to every membrane to represent the reactions.  Over the last
two decades, a considerable number of variants of P systems have been
introduced, inspired by various aspects of cellular life, or capturing
specific computing properties.  For comprehensive overviews we refer
the reader to \cite{imcs,handMC}.

Even though P systems resemble the organisation of a ``fundamental
unit'' of modern life, their use in representing actual biological
knowledge has historically been scarce.  Furthermore, one of the
salient examples of P systems in modelling are the works by the
Sevillan team (e.g.~\cite{CardonaCMPPHJPS2011,%
CardonaCMPHJPS2010,ColomerLMMPHJPSSVC2011,
ColomerMVP2014,GarciaQuismondoGRN2018,%
ValenciaCabreraGQJPSYP2013,ValenciaCabreraGPHJPRN2018}),
in which P systems represent ecosystems, an undeniably biological
structure, but far removed from the organisation of a cell.

To give more substance to this impression of underuse, we performed
a comparative bibliographic study of the literature using P systems to
represent any biological knowledge on the one hand, and on the other
hand the publications in the conference Computational Methods in
Systems Biology (e.g.~\cite{CMSB2021}) using Boolean networks to
represent biological knowledge.  A Boolean network is a set of Boolean
variables equipped with Boolean update functions, describing how to
compute the new value of the variables from their current values.
While Boolean networks represent well gene regulatory networks
(e.g.~\cite{Zanudo2018}), their structure arguably resembles less the
actual organisation of cellular processes.  Our study suggests
nonetheless that Boolean networks tend to be considerably more popular
than P systems for representing these processes.  We give the details
of this comparison in the appendix.

The main message of this paper is that the potential of P systems to
represent biological knowledge seems to remain relatively unexplored,
but that one can already rely on P systems as a flexible formal
framework providing powerful tools for studying other abstract
structures.  As an example, we show how to construct a P system
variant which naturally captures the semantics of sequentially
controlled Boolean networks.  In the future, this construction will
allow for more straightforward proofs of some properties of interest.

\medskip

This paper is structured as follows.  Sections~\ref{sec:preliminaries}
and~\ref{sec:seq-bn} recall the notions of P systems as well as Boolean 
networks, Boolean control networks, and sequential controllability.
Section~\ref{sec:boolp} introduces Boolean P systems.
Section~\ref{sec:quasimodes} introduces quasimodes to bridge between
the dynamics of Boolean networks and Boolean P~systems, and
Section~\ref{sec:boolp-bn} formally proves that Boolean P systems
capture Boolean networks.  Finally, Sections~\ref{sec:composition}
and~\ref{sec:boolp-seq} show how Boolean P systems explicitly embed
sequential controllability of Boolean networks.


\section{Preliminaries}
\label{sec:preliminaries}


To ensure unambiguous notation, in this section we briefly recall some
basic notions and concepts of formal language theory and membrane
computing.  For a detailed reference on both, we suggest~\cite{handMC}.

\medskip

For any alphabet $V$, $V^\circ$ is the set of multisets over $V$,  
and $V^*$ denotes the set of all strings over $V$. For any $u\in V^*$ 
and any $u\in V^\circ$, $|u|$ is the \emph{length} of the string $u$ and 
the number of elements in the multiset $u$, respectively. 
For $V^\circ$ and $V^*$ we use $\epsilon $ to denote the empty 
multiset and empty string, respectively.

We use $2^V$ to denote the set of all subsets of $V$ (the power set 
of $V$). Given two sets $A$ and $B$, by $B^A$ we denote the set of 
all functions $f : A \to B$.  

An \emph{indicator function} of a subset $U \subseteq V$ is the function
$i_U : V \to \{0,1\}$ with the property that $U = \{a \mid i_U(a) = 1\}$. 
In this paper, we will often use the same symbol to refer to a subset 
and to its indicator function.

A \emph{Boolean variable} is a variable which may only have values 
in the Boolean domain $\{0,1\}$.  


\subsection{P Systems}

\begin{definition}
 A {\em P system} is a construct
\[
\Pi = (O, T, \mu, w_1,\ldots,w_n, R_1,\ldots R_n, h_i, h_o),
\]
where $O$ is the alphabet of objects, $T\subseteq O$ is the alphabet
of terminal objects, $\mu$ is the membrane structure injectively
labelled by the numbers from $\{1,\ldots,n\}$ and usually given by
a sequence of correctly nested brackets, $w_i$ are the multisets
giving the initial contents of each membrane $i$ ($1\leq i\leq n$),
$R_i$ is the finite set of rules associated with membrane $i$
($1\leq i\leq n$), and $h_i$ and $h_o$ are the labels of the input and
the output membranes, respectively ($1\leq h_i\leq n$,
$1\leq h_o\leq n$).
\end{definition}

Quite often the rules associated with membranes are multiset rewriting
rules (or special cases of such rules). Multiset rewriting rules have
the form $u\to v$, with $u\in O^\circ\setminus \{\epsilon \}$ and 
$v\in O^\circ$. If $|u| = 1$, the rule $u\to v$ is called {\em non-cooperative};
otherwise it is called {\em cooperative}.  In {\em communication P systems}, 
rules are additionally allowed to send symbols to the
neighbouring membranes.  In this case, for rules in $R_i$,
$v\in (O\times \mathit{Tar}_i)^\circ$, where $\mathit{Tar}_i$ contains
the symbols $\mathit{out}$ (corresponding to sending the symbol to the
parent membrane), $\mathit{here}$ (indicating that the symbol should
be kept in membrane $i$), and $\mathit{in}_h$ (indicating that the
symbol should be sent into the child membrane $h$ of membrane $i$).
When writing out the multisets over $O \times \mathit{Tar}_i$, the
indication $\mathit{here}$ is often omitted.

In P systems, rules often are applied in a \emph{maximally parallel} way: 
in one derivation step, only a non-extendable multiset of rules can be
applied.  The rules are not allowed to consume the same instance of
a symbol twice, which creates competition for objects and may lead to
the P system choosing non-determinstically between the maximal
collections of rules applicable in one step. Yet rules may also be 
applied in a \emph{sequential} way, i.e. in every derivation step one 
rule which is applicable to the current configuration is carried out. 
Moreover, when any multiset of applicable rules may be applied, 
we speak of the \emph{asynchronous} derivation mode.

A computation of a P system is traditionally considered to be
a sequence of configurations it can successively visit by applying the 
applicable rules in the given derivation mode (maximally parallel, 
sequential, asynchronous), stopping at a halting configuration.  
A \emph{halting configuration} is a configuration in which no rule can be 
applied any more, in any membrane. The \emph{result of a computation} 
in a P system $\Pi$ as defined above is the contents of the output 
membrane $h_o$ projected over the terminal alphabet $T$.

\begin{example}
  Figure~\ref{fig:exampleP} shows the graphical representation of the
  P system formally given by
  \[
    \begin{array}{lcl}
      \Pi &=& (\{a,b,c\}, \{a,b\}, [_1[_2]_2]_1, R_1, R_2, 1, 2), \\
      R_1 &=& \emptyset, \\
      R_2 &=& \{c \to c(a,\mathit{out}),\, c \to c(b,\mathit{out}),\, c\to \epsilon\}.
    \end{array}
  \]

  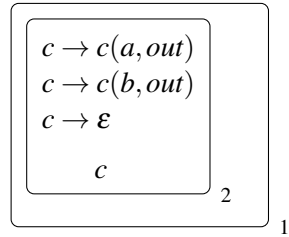
\begin{figure}[h]
    \centering
    \begin{tikzpicture}
      \tikzstyle membrane=[draw,rounded corners=1mm,inner sep=2mm]
      \node[membrane,align=left] (mem2) 
       {$c \to c(a,\mathit{out})$\\
        $c \to c(b,\mathit{out})$\\
        $c\to \epsilon$\\[2mm]
        \hspace{7mm}$c$};
      \node[right=0mm of mem2.south east] (lab2) {\scriptsize 2};
      \node[right=3mm of mem2] (skin content) {};
      \node[membrane,fit={(mem2) (lab2) (skin content)}] (skin) {};
      \node[right=0mm of skin.south east] {\scriptsize 1};
    \end{tikzpicture}
    \caption{An example of a simple P system.}
    \label{fig:exampleP}
  \end{figure}

In any derivation mode (maximally parallel, sequential, asynchronous), 
$\Pi$ may  apply one of the rules $c \to c(a,\mathit{out})$ or 
$c \to c(b,\mathit{out})$, thereby keeping the object $c$ in membrane~$2$ 
and at the same time sending out to membrane~$1$ one object $a$ or $b$, 
respectively. 

After $k$ such derivation steps, in membrane~$1$ a multiset 
$u\in \{a,b\}^{\circ}$ with $|u|=k$ has been obtained. Now applying 
the final rule $c\to \epsilon$, we obtain the halting configuration with 
no objects in membrane~$2$ and the multiset $u$ in membrane~$1$ 
as the result of the computation in~$\Pi$.\qed
\end{example}


\section{Sequential Controllability of Boolean Networks}
\label{sec:seq-bn}

In this section we briefly recall the definition of Boolean networks,
the extension of the formalism with control inputs, and the problem of
sequential controllability.  For a more in-depth coverage of these
definitions and problems, as well as the underlying biomedical
motivations, we refer the reader to~\cite{PardoID21}.


\subsection{Boolean Networks}
\label{sec:bn}

\begin{definition}
Let $X$ be a finite alphabet of Boolean variables.  A \emph{state} of 
these variables is any function $s$ in $\{0, 1\}^X$, i.e., $s : X \to \{0, 1\}$,
assigning a Boolean value to every single variable in $X$. By $S_X$ 
we denote the set  of all states $s$ in $\{0, 1\}^X$.

An \emph{update function} is a Boolean function computing
a Boolean value from a state: $f : s \to \{0,1\}$.  A \emph{Boolean
network} over $X$ is a function $F : S_X \to S_X$, in which the update
function for a variable $x \in X$ is computed as a projection of $F$:
$f_x(s) = F(s)_x$.
\end{definition}

A Boolean network $F$ computes trajectories on states by updating its
variables according to a (Bool\-ean) mode $M \subseteq 2^X$, defining
the variables which may be updated together in a step.
Typical examples of modes are the synchronous mode
$\mathit{syn} = \{X\}$ and the asynchronous mode
$\mathit{asyn} = \{\{x\} \mid x \in X\}$.  A trajectory $\tau$ of
a Boolean network under a given mode $M$ is any finite sequence of
states $\tau = (s_i)_{0 \leq i \leq n}$ such that $F$ can derive
$s_{i+1}$ from $s_i$ under the mode $M$.

An \emph{attractor} is a set of mutually reachable states $A \subseteq S_X$ of $F$ with the property that $F$ cannot escape from $A$.  Since the
set of states $S_X$ is finite, any run of a Boolean network, under any
mode, must end up in an attractor.  These are called the asymptotic
behaviors.

\begin{remark}
  These definitions are quite different from similar definitions
  generally used in P systems.  The asynchronous mode in Boolean
  networks only allows updating one variable at a time, while the
  asynchronous mode in P systems generally allows any combinations of
  updates.  Furthermore, no halting conditions are considered in
  Boolean networks, and the asymptotic behavior is often looked at as
  the important part of the dynamics. \qed
\end{remark}

\begin{example}\label{ex:bool}
  Consider the set of variables $X = \{x, y\}$ with the corresponding
  update functions $f_x(x, y) = \bar x \wedge y$ and
  $f_y(x, y) = x \wedge \bar y$.  Figure~\ref{fig:example-bool} shows
  the possible state transitions of this network under the synchronous
  and the asynchronous modes.  The states are represented as pairs of
  binary digits, e.g. $01$ stands for the state in which $x = 0$ and
  $y = 1$.

  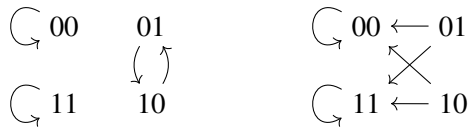
\begin{figure}[h]
    \centering
    \begin{tikzpicture}[node distance=5mm]
      \node (00) {$00$};
      \node[right=of 00] (01) {$01$};
      \node[below=of 01] (10) {$10$};
      \node[below=of 00] (11) {$11$};

      \draw[->] (00) to[out=150,in=-150,looseness=4] (00);
      \draw[->] (11) to[out=150,in=-150,looseness=4] (11);
      \draw[->] (01) to[bend right] (10);
      \draw[->] (10) to[bend right] (01);
    \end{tikzpicture}
    \hspace{15mm}
    \begin{tikzpicture} [node distance=5mm]
      \node (00) {$00$};
      \node[right=of 00] (01) {$01$};
      \node[below=of 01] (10) {$10$};
      \node[below=of 00] (11) {$11$};

      \draw[->] (00) to[out=150,in=-150,looseness=4] (00);
      \draw[->] (11) to[out=150,in=-150,looseness=4] (11);

      \draw[->] (01) to (11);
      \draw[->] (01) to (00);

      \draw[->] (10) to (11);
      \draw[->] (10) to (00);
    \end{tikzpicture}
    \caption{The synchnronous (left) and the asynchnronous (right)
      dynamics of the Boolean network in Example~\ref{ex:bool}.}
    \label{fig:example-bool}
  \end{figure}

  We notice that, under the synchronous mode, this network exhibits
  two kinds of behaviors. If initialized in the state $00$ or $11$,
  it will stay in the initial state forever---these two are stable states.
  If it is initialized in any one of the states $01$ or $10$, it will
  oscillate between them. The behavior of the network therefore is
  deterministic under the synchronous update mode.

  The state transitions are quite different under the asynchronous
  mode, under which only one variable may be updated at a time.
  While the states $00$ and $11$ remain stable, two possible
  transitions are now from the states $01$ and $10$, and
  there are no transitions leading from $01$ to $10$ or vice
  versa. \qed
\end{example}


\subsection{Boolean Control Networks}
\label{sec:bcn}

Boolean networks are often used to represent biological networks in
the presence of external perturbations: environmental hazards, drug
treatments, etc. (e.g.,~\cite{Barabasi2011,BianeD19,PardoID21}). To
represent network reprogramming, an extension of Boolean networks can
be considered: Boolean control networks (BCN)~\cite{BianeD19}.
Informally, a BCN is a parameterized Boolean network template;
assigning a Boolean value to every single one of its parameters yields
a Boolean network.

Formally, a Boolean control network is a function
$F_U : S_U \to (S_X \to S_X)$, where the elements of $U$,
$U \cap X = \emptyset$, are called the control inputs.  To every
valuation of control inputs, $F_U$ associates a Boolean network.
A control $\mu$ of $F_U$ is any Boolean assignment to the control
inputs: $\mu : U \to \{0, 1\}$.

\smallskip

While this definition of BCNs is very general, in practice one
restricts the impact the control inputs may have on the BCN to some
biologically relevant classes.  One particularly useful class are
freeze perturbations, in which a variable in $X$ is temporarily frozen
to 0 or to 1, independently of its normal update function.

When Boolean update functions are written as propositional formulae,
freeze control inputs can be written directly in the formulae of the
update functions. For example, consider a Boolean network $F$ over
$X = \{x_1, x_2\}$ with the update functions $f_1 = x_1 \wedge x_2$
and $f_2 = x_2$.  To allow for freezing $x_1$, we introduce the
control variables $U = \{u_1^0, u_1^1\}$ into the Boolean formula of
$f_1$ in the following way:
$f'_1 = (x_1 \wedge x_2) \wedge u_1^0 \vee \bar u_1^1$.
Setting $u_1^0$ to 0 and $u_1^1$ to 1 freezes $x_1$ to 0,
independently of the values of $x_1$ and $x_2$.  Symmetrically,
setting $u_1^1$ to 0 and $u_1^0$ to 1 (or 0) freezes $x_1$ to 1.


\subsection{Sequential Controllability of Boolean Control Networks}
\label{sec:seq-bcn}

In many situations, perturbations of biological networks do not happen
once, but rather accumulate or evolve over
time~\cite{Fearon1990,Lee2012,PardoID21}.  In the language of Boolean
control networks, this corresponds to considering sequences of
controls $(\mu_1, \dots, \mu_n)$.  More precisely, take a BCN $F_U$
with the variables $X$ and the control inputs $U$, as well as
a sequence of controls $\mu_{[n]} = (\mu_1, \dots, \mu_n)$,
$\mu_i : U \to \{0, 1\} \in S_U$.  This gives rise to a sequence of
Boolean networks $(F_U(\mu_1), \dots, F_U(\mu_n))$.  Fix a mode $M$
and consider a sequence of trajectories $(\tau_1, \dots, \tau_n)$ of
these Boolean networks.  Such a sequence is an evolution of $F_U$
under the sequence of controls $\mu_{[n]}$ if the last state of every
$\tau_i$ is the first state of $\tau_{i+1}$.  In this case we can
speak of the trajectory of the BCN $F_U$ under the control sequence
$\mu_{[n]}$ as the concatenation of the individual trajectories
$\tau_i$, in which the last state of every single $\tau_i$ is glued
together with the first state of $\tau_{i+1}$.

The problem of inference of control sequences (the CoFaSe problem) was
extensively studied in~\cite{PardoID21}.  Given the 3-tuple
$(F_U, S_\alpha, S_\omega)$, where $F_U$ is a BCN, $S_\alpha$ is a set
of starting states, and $S_\omega$ is a set of target states, the
CoFaSe problem consists in inferring a control sequence driving $F_U$
from any state in $S_\alpha$ to any state in $S_\omega$.  Deciding the
existence of such a sequence is PSPACE-hard.

\begin{example}\label{ex:bcn}
  While the framework of Boolean control networks allows for
  considering arbitrary kinds of control actions, it has been
  extensively used (e.g.~\cite{PardoID21}) for capturing freezing,
  i.e. setting and maintaining specific variables at specific values.
  These actions mean to model gene knock-ins and knock-outs.

  Consider again the Boolean network from Example~\ref{ex:bool}, with
  $X = \{x, y\}$ and the update functions $f_x = \bar x \wedge y$ and
  $f_y = x \wedge \bar y$.  A convenient way to express freezing
  controls is by explicitly including the control inputs into the
  update functions in the following way:
  \[
    \renewcommand{\arraystretch}{1.2}
    \begin{array}{lcl}
      f'_x &=& (\bar x \wedge y) \wedge u_x^0 \vee u_x^1, \\
      f'_y &=& (x \wedge \bar y) \wedge u_y^0 \vee u_y^1. \\
    \end{array}
  \]
  Notice that setting $u_x^0$ to 0 essentially sets $f_x' = 0$, and
  setting $u_x^1$ to 0 essentially sets $f_x' = 1$, independently of
  the actual value of $x$ or $y$.

  Consider now the following 3 controls:
  \[
    \renewcommand{\arraystretch}{1.3}
    \begin{array}[lcl]{lcl}
      \mu_1 &=& \{u_x^0 \leftarrow 0, u_x ^1 \leftarrow 0, u_y^0 \leftarrow 0, u_y^1 \leftarrow 0\}, \\
      \mu_2 &=& \{\underline{u_x^0 \leftarrow 1}, u_x ^1 \leftarrow 0, u_y^0 \leftarrow 0, u_y^1 \leftarrow 0\}, \\
      \mu_3 &=& \{u_x^0 \leftarrow 0, u_x ^1 \leftarrow 0, u_y^0 \leftarrow 0, \underline{u_y^1 \leftarrow 1}\}. \\
    \end{array}
  \]
  Informally $\mu_1$ does not freeze any variables, $\mu_2$ freezes
  $x$ to 0, and $\mu_3$ freezes $y$ to~1.  Consider now the BCN $F_U$
  with the variables $X = \{x,y\}$ and the controlled update functions
  $f_x'$ and $f_y'$.  Fix the synchronous update mode.  A trajectory
  of this BCN under the control $\mu_1$---i.e. a trajectory of
  $F_U(\mu_1)$---is $\tau_1 : 01 \to 10 \to 01$.  A trajectory of
  $F_U(\mu_2)$ is $\tau_2 : 01 \to 00 \to 00$; remark that $00$ is
  still a stable state of $F_U(\mu_2)$.  A trajectory of $F_U(\mu_3)$
  is $\tau_3 : 00 \to 01 \to 11$.  We can now glue together the
  trajectories $\tau_1$, $\tau_2$, and $\tau_3$ by identifying their
  respective ending and starting states, and we will obtain the
  following trajectory of the BCN $F_U$ under the control sequence
  $\mu_{[3]} = (\mu_1, \mu_2, \mu_3)$:
  \[
    \tau : 01 \to 10 \to 01 \to 00 \to 00 \to 01 \to 11.
  \]
  It follows from this construction that $\mu_{[3]}$ is a solution for
  the CoFaSe problem $(F_U, \{01\}, \{11\})$.  Remark that $11$ is not
  reachable from $01$ in the uncontrolled case, as
  Figure~\ref{fig:example-bool} illustrates. \qed
\end{example}


\section{Boolean P Systems}
\label{sec:boolp}

In this section we introduce a new variant of P systems---Boolean
P systems---tailored specifically to capture sequential control of
Boolean networks with as little descriptional overhead as possible.
Rather than trying to be faithful to the original model as recalled in
Section~\ref{sec:preliminaries}, we here invoke the intrinsic
flexibility of the domain to design a variant fitting to our specific
use case.

We construct Boolean P systems as set rewriting systems.  A Boolean
state $s : X \to \{0, 1\}$ will be represented as the subset of $X$
obtained by considering $s$ as an indicator function:
$\{x \in X \mid s(x) = 1 \}$.  By abuse of notation, we will sometimes
use the symbol $s$ to refer both to the Boolean state and to the
corresponding subset of $X$.

A Boolean P system is a construct
\[
  \Pi = (V, R),
\]
where $V$ is the alphabet of symbols, and $R$ is a set of rewriting
rules with guards.  A rule $r \in R$ is of the form
\[
  r : A \to B \mid \varphi,
\]
where $A, B \subseteq X$ and $\varphi$ is the guard---a propositional
formula with variables from $V$.  The rule $r$ is applicable to a set
$W \subseteq V$ if $A \subseteq W$ and $W \in \varphi$, where by abuse
of notation we use the same symbol $\varphi$ to indicate the set of
subsets of $V$ which satisfy $\varphi$.  Formally, for
$W \subseteq V$, by $\varphi(W)$ we denote the truth value of the
formula obtained by replacing all variables appearing in $W$ by
1 in $\varphi$, and by 0 all variables from $V \setminus W$.  Then the
set of subsets satisfying $\varphi$ is
$\varphi = \{W \subseteq V \mid \varphi(W) \equiv \mathbf{1}\}$, 
where $\mathbf{1}$ is the Boolean tautology.

Applying the rule $r : A \to B \mid \varphi$ to a set $W$ results in
the set $(W \setminus A) \cup B$.  Applying a set of separately
applicable rules $\{r_i : A_i \to B_i \mid \varphi_i\}$ to $W$ results
in the new set
\[
  \left(W \setminus \bigcup_i A_i \right) \cup \bigcup_i B_i.
\]
Note how this definition excludes competition between the rules, as
only individual applicability is checked.  Further note that applying
a rule multiple times to the same configuration has exactly the same
effect as applying it once.

In P systems, the set of multisets of rules of $\Pi$ applicable to
a given configuration $W$ is usually denoted by
$\mathit{Appl}(\Pi, W)$~\cite{FreundV2007}.  Since in Boolean
P systems multiple applications of rules need not be considered, we
will only look at the set of \emph{sets} of rules applicable to
a given configuration $W$ of a Boolean P system $\Pi = (V, R)$, and
use the same notation $\mathit{Appl}(\Pi, W)$.  A mode $M$ of $\Pi$
will then be a function assigning to any configuration $W$ of $\Pi$
a set of sets of rules applicable to $W$, i.e.,
\[
  M : 2^V \to 2^R \text{ such that } M(W) \subseteq \mathit{Appl}(\Pi, W).
\]
If $|M(W)| \leq 1$ for any $W \subseteq V$, the mode $M$ is called
deterministic.  Otherwise it is called non-deterministic.

An evolution of $\Pi$ under the mode $M$ is a sequence of states
$(W_i)_{0 \leq i \leq k}$ with the property that $W_{i+1}$ is obtained
from $W_i$ by applying one of the sets of rules $R' \in M(W_i)$
prescribed by the mode $M$ in the state $W_i$.  This is usually
written as $W_i \overset{R'}{\longrightarrow} W_{i+1}$.  If no rules
are applicable to the state $W_k$, $W_k$ is called \emph{halting state}, 
and $(W_i)_{0 \leq i \leq k}$ is called a halting evolution.

Finally, we remark that the starting state is not part of this
definition of a Boolean P system.  We make this choice to better
parallel the way in which Boolean networks are defined.

\begin{example}
  Take $V = \{a, b\}$ and consider the following rules
  $r_1 : \{a, b\} \to \{a\} \mid \mathbf{1}$ and
  $r_2 : \{a\} \to \emptyset \mid \bar b$, where $\mathbf{1}$ is the
  Boolean tautology.  Construct the Boolean P system
  $\Pi = (V, \{r_1, r_2\})$.  Informally, $r_1$ removes $b$ from
  a configuration which contains $a$ and $b$, and $r_2$ removes $a$
  from the configuration which does not already contain $b$.
  A possible trajectory of $\Pi$ under the maximally parallel
  mode---which applies non-extendable applicable sets of rules---is
  $\{a, b\} \to \{a\} \to \emptyset$.  Note that only $r_1$ is
  applicable in the first step, since $r_2$ requires the configuration
  to not contain $b$. \qed
\end{example}

\begin{remark}\label{rem:boolp-rs}
  Boolean P systems as defined here are very close to other set
  rewriting formalisms, and in particular to reaction
  systems~\cite{EhrenfeuchtR2007}.  A reaction system $\mathcal A$
  over a set of species $S$ is a set of reactions (rules) of the form
  $a : (R_a, I_a, P_a)$, in which $R_a \subseteq S$ is called the set
  of reactants, $I_a \subseteq S$ the set of inhibitors, and
  $P_a \subseteq S$ the set of products.  For $a$ to be applicable to
  a set $W$, it must hold that $R_a \subseteq W$ and
  $I_a \cap W = \emptyset$.  Applying such a reaction to $W$ yields
  $P_a$, i.e., the species which are not explicitly sustained by the
  reactions disappear.

  We claim that despite their apparent similarity and tight
  relationship with Boolean functions, reaction systems are not such a
  good fit for reasoning about Boolean networks as Boolean
  P systems.  In particular:
  \begin{enumerate}
  \item Reaction systems lack modes and therefore non-determinism,
    which may appear in Boolean networks under the asynchronous
    Boolean mode.
  \item The rule applicability condition is more powerful in Boolean
    P systems, and closer to Boolean functions than in
    reaction systems.
  \item Symbols in reaction systems disappear unless sustained by
    a rule, which represents the degradation of species in
    biochemistry, but which makes reaction systems harder to use to
    directly reason about Boolean networks.
  \end{enumerate}
  We recall that our main goal behind introducing Boolean P systems is
  reasoning about Boolean networks in a more expressive framework.
  This means that zero-overhead representation of concepts from
  Boolean networks is paramount. \qed
\end{remark}

\begin{remark}\label{rem:rs-control}
  Reaction systems~\cite{EhrenfeuchtR2007} are intrinsically
  interesting for discussing controllability, because they are defined
  as open systems from the start, via the explicit introduction of
  context.  We refer to~\cite{IvanovP20} for an in-depth discussion of
  controllability of reaction systems. \qed
\end{remark}


\section{Quasimodes}
\label{sec:quasimodes}

An update function in a Boolean network can always be computed, but
a rule in a Boolean P system need not always be applicable.  This is
the reason behind the difference in the way modes are defined in the
two formalisms: in Boolean networks a mode is essentially a set of
subsets of update functions, while in Boolean P~systems a mode is
a function incorporating applicability checks.  This means in
particular that Boolean network modes are not directly transposable to
Boolean P systems.

To better bridge the two different notions of modes, we introduce
quasimodes.  A \emph{quasimode} $\tilde M$ of a P system
$\Pi = (V, R)$ is any set of sets of rules: $\tilde M \subseteq 2^R$.
The mode $M$ corresponding to the quasimode $\tilde M$ is derived in
the following way:
\[
  M(W) = \tilde M \cap \mathit{Appl}(\Pi, W).
\]
Given a configuration $W$ of $\Pi$, $M$ picks only those sets of rules
from $\tilde M$ which are also applicable to $W$.  Thus, instead of
explicitly giving the rules to be applied to a given configuration of
a P system $W$, a quasimode advises the rules to be applied.

In the rest of the paper, we will say ``evolution of $\Pi$ under the
quasimode~$\tilde M$'' to mean ``evolution of $\Pi$ under the mode
derived from the quasimode $\tilde M$''.


\section{Boolean P Systems Capture Boolean Networks}
\label{sec:boolp-bn}


Consider a Boolean network $F$ over the set of variables $X$, and take
a variable $x \in X$ with its corresponding update function $f_x$.
The update function $f_x$ can be simulated by two Boolean P systems
rules: the rules corresponding to setting $x$ to 1, i.e. introducing
$x$ into the configuration, and the rules corresponding to setting $x$
to 0, i.e. erasing $x$ from the configuration:
\[
  R_x = \{\;\; \emptyset \to \{x\} \mid f_x, \;\; \{x\} \to \emptyset
  \mid \neg f_x \;\;\}.
\]
Now consider the following Boolean P system:
\[
  \Pi(F) = \left(X, \bigcup_{x \in X} R_x \right).
\]
We claim that $\Pi(F)$ faithfully simulates $F$.

\begin{theorem}\label{thm:bp-bn}
  Take a Boolean network $F$ and a Boolean mode $M$.  Then the Boolean
  P system $\Pi(F)$ constructed as above and working under the
  quasimode
  $\tilde M = \left\{\bigcup_{x \in m} R_x \mid m \in M\right\}$
  faithfully simulates $F$: for any evolution of $F$ under $M$ there
  exists an equivalent evolution of $\Pi(F)$ under $\tilde M$, and
  conversely, for any evolution of $\Pi(F)$ under $\tilde M$ there
  exists an equivalent evolution of $F$ under~$M$.
\end{theorem}

\begin{proof}
  Consider two arbitrary states $s$ and $s'$ of $F$ such that $s'$ is
  reachable from $s$ by the update prescribed by an element $m \in M$.
 Now consider the subsets of variables $W, W' \subseteq X$ defined by
  $s$ and $s'$ taken as respective indicator functions.  It follows
  from the construction of $\tilde M$ that it contains an element
  $\tilde m$ including the update rules for all the variables of $m$:
  $\tilde m = \bigcup_{x \in m} R_x$.  Therefore, $\Pi(F)$ can derive
  $W'$ from $W$ under the quasimode $\tilde M$.

  Conversely, consider two subsets of variables $W, W' \subseteq X$
  such that $\Pi(F)$ can derive $W'$ from $W$ under the update
  prescribed by an element $\tilde m \in \tilde M$.  By construction
  of $\tilde M$, there exists a subset $m \subseteq X$ such that
  $\tilde m = \bigcup_{x \in m} R_x$.  Now take the indicator
  functions $s, s' : X \to \{0, 1\}$ describing $W$ and $W'$
  respectively.  Then $F$ can derive $s'$ from $s$ by updating the
  variables in $m$.

  We conclude that the transitions of $\Pi(F)$ exactly correspond to
  the transitions of $F$, which proves the statement of the
  theorem.
\end{proof}

The above proof stresses the original motivation behind the
introduction of Boolean P systems as a framework for direct and easy
generalization of Boolean networks: Boolean P systems were designed to
make the simulation of Boolean networks as easy as possible.

\begin{remark}
  Incidentally, Boolean P systems also capture reaction systems (see
  also Remarks~\ref{rem:boolp-rs} and~\ref{rem:rs-control}).  Indeed,
  consider a reaction $a = (R_a, I_a, P_a)$ with the reactants $R_a$,
  inhibitors $I_a$, and products $P_a$.  It can be directly simulated
  by the Boolean P system rule $\emptyset \to P_a \mid \varphi_a$,
  where
  $\varphi_a = \bigwedge_{x \in R_a} x \wedge \bigwedge_{y \in I_a}
  \bar y$.  The degradation of the species in reaction systems is
  simulated by adding a rule $x \to \emptyset \mid \mathbf{1}$ for
  every species $x$, where $\mathbf{1}$ is the Boolean tautology. \qed
\end{remark}


\section{Composition of Boolean P Systems}
\label{sec:composition}


In this section, we define the composition of Boolean P systems in the
spirit of automata theory.  Consider two Boolean P systems
$\Pi_1 = (V_1, R_1)$ and $\Pi_2 = (V_2, R_2)$.  We will call the union
of $\Pi_1$ and $\Pi_2$ the Boolean P system
$\Pi_1 \cup \Pi_2 = (V_1 \cup V_2, R_1 \cup R_2)$.  Note that the
alphabets $V_1$ and $V_2$, as well as the rules $R_1$ and $R_2$ are
not necessarily disjoint.

To talk about the evolution of $\Pi_1 \cup \Pi_2$, we first define
a variant of Cartesian product of two sets of sets $A$ and $B$, which
consists in taking the union of the elements of the pairs:
$A \dottimes B = \{a \cup b \mid a \in A, b \in B\}$.  We remark
now that
\[
  \forall W\subseteq V_1 \cup V_2 : \mathit{Appl}(\Pi_1 \cup \Pi_2, W)
  = \mathit{Appl}(\Pi_1, W) \dottimes \mathit{Appl}(\Pi_2, W).
\]
Indeed, since the rules of Boolean P systems do not compete for
resources among them, the applicability of any individual rule is
independent of the applicability of the other rules.  Therefore, the
applicability of a set of rules of $\Pi_1$ to a configuration $W$ is
independent of the applicability of a set of rules of $\Pi_2$ to~$W$.

For a mode $M_1$ of $\Pi_1$ and a mode $M_2$ of $\Pi_2$, we define
their product as follows:
\[
  (M_1 \times M_2)(W) = M_1(W) \dottimes M_2(W).
\]
The union of Boolean P systems $\Pi_1 \cup \Pi_2$ together with the
product mode $M_1 \times M_2$ implements parallel composition of the
two P systems.  In particular, if the alphabets of $\Pi_1$ and $\Pi_2$
are disjoint, the projection of any evolution of $\Pi_1 \cup \Pi_2$
under the mode $M_1 \times M_2$ on the alphabet $V_1$ will yield
a valid evolution of $\Pi_1$ under $M_1$ (modulo some repeated
states), while the projection on $V_2$ will yield a valid evolution of
$\Pi_2$ under the mode $M_2$ (modulo some repeated states).  Note this
property may not be true if the two alphabets intersect
$V_1 \cap V_2 \neq \emptyset$.

Quasimodes fit naturally with the composition of modes, as the
following lemma shows.

\begin{lemma}
  If the mode $M_1$ can be derived from the quasimode $\tilde M_1$ and
  $M_2$ from the quasimode $\tilde M_2$, then the product mode
  $M_1 \times M_2$ can be derived from
  $\tilde M_1 \dottimes \tilde M_2$:

  \begin{center}
    \begin{tikzpicture}[node distance=7mm and 6mm]
      \node (M1-x-M2) {$M_1 \times M_2$};
      \node[above=of M1-x-M2] (tM1-x-tM2) {$\tilde M_1 \dottimes \tilde M_2$};
      \node[yshift=-3mm,base left=of tM1-x-tM2] (tM1) {$\tilde M_1$};
      \node[yshift=-3mm,base right=of tM1-x-tM2] (tM2) {$\tilde M_2$};
      \node[yshift=-1mm,below=of tM1] (M1) {$M_1$};
      \node[yshift=-1mm,below=of tM2] (M2) {$M_2$};

      \draw[->,densely dashed] (tM1) -- (M1);
      \draw[->,densely dashed] (tM2) -- (M2);
      \draw[->,densely dashed] (tM1-x-tM2) -- (M1-x-M2);

      \draw[<-] (M1) -- (M1-x-M2);
      \draw[<-] (M2) -- (M1-x-M2);
      \draw[<-] (tM1) -- (tM1-x-tM2);
      \draw[<-] (tM2) -- (tM1-x-tM2);
    \end{tikzpicture}
  \end{center}
  where a dashed arrow \tikz\draw[->,densely dashed] (0,0) -- (5mm,0);
  from a quasimode to a mode indicates that the mode is derived from
  the quasimode, and the arrows \tikz\draw[->] (0,0) -- (5mm,0); are
  the respective projections.
\end{lemma}
\begin{proof}
  Pick a state $W \subseteq X$ and recall that the mode $M_{12}$
  derived from $\tilde M_1 \dottimes \tilde M_2$ is defined as
  follows:
  \[
    M_{12}(W) = \left(\tilde M_1 \dottimes \tilde M_2\right) \cap
    \mathit{Appl}(\Pi, W).
  \]

  Consider an arbitrary element $m_{12} \in M_{12}(W)$ and remark that
  it can be seen as a union $m = m_1 \cup m_2$ where $m_1$ is a subset
  of applicable rules with the property that $m_1 \in \tilde M_1$, and
  $m_2$ is a subset of applicable rules with the property that
  $m_2 \in \tilde M_2$.  Thus
  $m_1 \in \tilde M_1 \cap \mathit{Appl}(\Pi,W)$ and
  $m_2 \in \tilde M_2 \cap \mathit{Appl}(\Pi,W)$, implying that
  \[
    M_{12}(W) \subseteq
    \left(\tilde M_1 \cap \mathit{Appl}(\Pi,W)\right)
    \dottimes
    \left(\tilde M_2 \cap \mathit{Appl}(\Pi,W)\right).
  \]

  On the other hand, consider arbitrary
  $m_1 \in \tilde M_1 \cap \mathit{Appl}(\Pi,W)$ and arbitrary
  $m_2 \in \tilde M_2 \cap \mathit{Appl}(\Pi,W)$.  By definition of
  $\dottimes$,
  $m_1 \cup m_2 \in \tilde M_1 \dottimes \tilde M_2$.  Remark that
  every rule in $m_1$ and $m_2$ is individually applicable, meaning
  that they are also applicable together and that
  $m_1 \cup m_2 \in \mathit{Appl}(\Pi, W)$.  Combining this
  observation with the reasoning from the previous paragraph we
  finally derive:
  \[
    M_{12}(W) =
    \left(\tilde M_1 \cap \mathit{Appl}(\Pi,W)\right)
    \dottimes
    \left(\tilde M_2 \cap \mathit{Appl}(\Pi,W)\right)
    = M_1(W) \dottimes M_2(W),
  \]
  which implies that $M_{12} = M_1 \times M_2$ and concludes the
  proof.
\end{proof}


\section{Boolean P Systems Capture Sequential Controllability}
\label{sec:boolp-seq}


Underlying sequential controllability of Boolean control networks
(Section~\ref{sec:seq-bcn}) is the implicit presence of a master
dynamical system emitting the control inputs of the network and
thereby driving it.  This master system is external with respect to
the controlled BCN.  The framework of Boolean P systems is
sufficiently general to capture both the master system and the
controlled BCN in a single homogeneous formalism.  In this section, we
show how to construct such Boolean P systems for dealing with
questions of controllability.

Any BCN $F_U : S_U \to (S_X \to S_X)$ can be written as a set of
propositional formulae over $X \cup U$.  Indeed, any control
$\mu \in S_U$ can be translated into the conjuction
$\bigwedge_{u \in \mu} u \wedge \bigwedge_{v \in U \setminus \mu} \bar
v$.  Now fix an $x \in X$ and consider the formula
\begin{equation}
  \bigvee_{\mu \in S_U} \mu \wedge F(\mu)_x,
  \label{eq:controlled-update-function}
\end{equation}
in which $\mu$ enumerates all the conjuctions corresponding to the
controls in $S_U$ and $F(\mu)_x$ is the propositional formula of the
update function which $F$ associates to $x$ under the control $\mu$.
With the formulae (\ref{eq:controlled-update-function}), we can
translate any BCN $F_U : S_U \to (S_X \to S_X)$ into
$F' : S_{X \cup U} \to S_X$ and use the set $R_x$ from
Section~\ref{sec:boolp-bn} to further translate the individual
components of $F'$ to pairs of Boolean P system rules.  Denote
$\Pi = (X \cup U, R)$ the Boolean P system whose set of rules is
precisely the union of the sets $R_x$ mentioned above. Finally,
construct the Boolean P~system $\Pi_U(U, R_U)$ with the following
rules whose guards are always true:
\[
  \renewcommand{\arraystretch}{1.3}
  \begin{array}{lcl}
    R_U &=& R_U^0 \cup R_U^1,\\[1mm]
    R_U^0 &=& \{\;\,\{u\} \to \emptyset \mid \mathbf{1} \;\, \mid u \in U\,\},\\
    R_U^1 &=& \{\;\,\emptyset \to \{u\} \mid \mathbf{1} \;\, \mid u \in U\,\}. \\
  \end{array}
\]

Suppose now that the original BCN $F_U$ runs under the mode $M$, and
consider the corresponding quasimode
$\tilde M = \left\{\bigcup_{x\in m}R_x \mid m \in M\right\}$, as well
as the quasimode
\[
  \tilde M_{U} =  \{R_U^0\} \dottimes 2^{R_U^1}.
\]
Every element of $\tilde M_U$ is a union of $R_U^0$ and a subset of
$R_U^1$. We claim that the Boolean P system $\Pi \cup \Pi_U$ running
under the quasimode $\tilde M \dottimes \tilde M_{U}$ faithfully
simulates the BCN $F_U$ running under the mode $M$.  The following
theorem formalizes this claim.

\begin{theorem}\label{thm:bp-seq}
  Consider a BCN $F_U$ running under the mode $M$.  Then the Boolean
  P system $\Pi \cup \Pi_U$ constructed as above and running under the
  quasimode $\tilde M \dottimes \tilde M_U$ faithfully simulates
  $F_U$:
  \begin{enumerate}
  \item For any evolution of $F_U$ under $M$ there exists an
    equivalent evolution of $\Pi \cup \Pi_U$ under
    $\tilde M \dottimes \tilde M_U$.
  \item For any evolution of $\Pi \cup \Pi_U$ under
    $\tilde M \dottimes \tilde M_U$ there exists an equivalent
    evolution of $F_U$ under $M$.
  \end{enumerate}
\end{theorem}
\begin{proof}
  \emph{(1)}\hspace{1ex} Consider two states $s, s' \in S_X$ and
  a control $\mu \in S_U$ such that $F_U(\mu)$ reaches $s'$ from $s$
  in one step.  Take $W, W' \subseteq X$ and $W_U \subseteq U$ by
  respectively taking $s$, $s'$, and $\mu$ as indicator functions.
  Then, as in Theorem \ref{thm:bp-bn}, there exists an
  $\tilde m \in \tilde M$ such that $\Pi$ reaches $W' \cup W_U$ from
  $W \cup W_U$ in one step.  This follows directly from the
  construction of the rules in $\Pi$ and from the fact that $W_U$
  contains exactly the symbols corresponding to the control inputs
  activated by~$\mu$.

  Now take $\tilde M \dottimes \tilde M_U$ and remark that it
  contains an element $\tilde m \cup \tilde m_U$, where
  $\tilde m_U = \tilde m_U^1 \cup R_U^0$ and
  $\tilde m_U^1 \subseteq R_U^1$.  Under this element
  $\tilde m \cup \tilde m_U$, $\Pi \cup \Pi_U$ reaches a state
  $W' \cup W_U'$ from $W \cup W_U$ in one step, where $W_U'$ contains
  the symbols from $U$ introduced by the rules selected by
  $\tilde m_U^1$.  Further note that all the elements of $W_U$ are
  always erased by the rules $R_U^0$, but may be reintroduced by
  $m_U^1$.

  Suppose that $F_U(\mu)$ reaches $s'$ from $s$ in multiple steps.
  Then $\Pi$ reaches $W' \cup W_U$ from $W \cup W_U$ in the same
  number of steps, provided that $\tilde m_U^1$ is always chosen such
  that the rules it activates reintroduce exactly the subset $W_U$.
  If $F_U$ reaches $s'$ from $s$ in multiple steps, but the control
  evolves as well, it suffices to choose $\tilde m_U^1$ such that it
  introduces the correct control inputs before each step.  Finally,
  the control $\mu_0$ applied in the first step of a trajectory of
  $F_U$ must be introduced by setting the starting state of
  $\Pi \cup \Pi_U$ to $W \cup W_U^0$, where $W$ corresponds to the
  initial state of the trajectory of $F_U$.

  \medskip

  \noindent
  \emph{(2)}\hspace{1ex} The converse construction is symmetric.
  A state $W \cup W_U$ of $\Pi \cup \Pi_U$ is translated into the
  state $s \in S_X$ and the control $\mu \in S_U$ corresponding to
  $W_U$.  A step of $\Pi \cup \Pi_U$ under $\tilde m \cup \tilde m_U$
  is translated to applying $\mu$ to $F_U$ and updating the variables
  corresponding to the rules activated by $\tilde m$.  In this way,
  for any trajectory of $\Pi \cup \Pi_U$ under the quasimode
  $\tilde M \dottimes \tilde M_U$ there exists a corresponding
  trajectory in the controlled dynamics of $F_U$.
\end{proof}

The component $\Pi_U$ in the composite P system of
Theorem~\ref{thm:bp-seq} is an explicit implementation of the master
dynamical system driving the evolution of the controlled system $\Pi$.
The setting of this theorem captures the situation in which the
control can change at any moment, but $\Pi_U$ can be designed to
implement other kinds of control sequences.  We give the construction
ideas for the kinds of sequences introduced in~\cite{PardoID21}:
\begin{itemize}[itemsep=1.5mm]
\item \emph{Total Control Sequence (TCS):} all controllable variables
  are controlled.

  \smallskip

  The quasimode of $\Pi_U$ will be correspondingly defined to always
  freeze the controlled variables:
  $\tilde M_U = \{R_U^0\} \dottimes 2^{P_U^1}$, where
  $P_U^1 \subseteq R_U^1$ with the property that for every $x_i \in X$
  every set $p \in P_U^1$ either introduces $u_i^0$ or $u_i^1$, but
  not both.
\item \emph{Abiding Control Sequence (ACS):} once controlled,
  a variable stays controlled forever, but its value may change.

  \smallskip

  The rules of $\Pi_U$ will be constructed to never erase the control
  symbols which have already been introduced, but will be allowed to
  change the value to which the corresponding controlled variable will
  be frozen: $R_U = R_U^1 \cup P_U$, with the new set of rules defined
  as follows:
  \[
    P_U = \left\{\;\{u_i^a\} \to \{u_i^b\} \mid \mathbf{1} \, \mid x_i
      \in X, \, a, b \in \{0,1\}\right\}.
  \]
  The P system $\Pi_U$ will be able to rewrite some of the control
  symbols, or to introduce new control symbols:
  $\tilde M_U = 2^{R_U}$.
\end{itemize}


\section{Conclusion}


The motivation of this work stems from the relative underuse of
P systems in representing biological knowledge, in spite of its
obvious biological inspiration.  To informally confirm this intuition
of underuse, we established a state of the art comparing the numbers
of publications using P systems and Boolean networks to represent any
kind of biological knowledge.  Our conclusion is that Boolean networks
tend to be more popular in this line of research.  We speculate that
the reason behind this relative popularity of Boolean networks is the
greater simplicity of the formalism and original interest on the part
of the biological community.

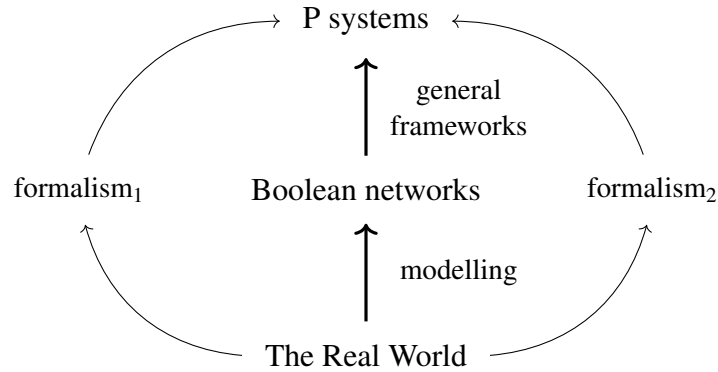
\begin{figure}
  \centering
  \begin{tikzpicture}[node distance=13mm and 8mm]
    \node[inner sep=3mm] (real) {\large The Real World};
    \node[inner sep=3mm,above=of real] (boolean) {\large Boolean networks};
    \node[inner sep=3mm,above=of boolean] (psystems) {\large P systems};
    \node[inner sep=3mm,left=of boolean] (l-ellipsis) {formalism$_1$};
    \node[inner sep=3mm,right=of boolean] (r-ellipsis) {formalism$_2$};

    \newcommand{\bendamount}{35}
    \draw[-To,very thick] (real) -- (boolean);
    \draw[-To,very thick] (boolean) -- (psystems);
    \draw[-To] (real.west) to[bend left=\bendamount] (l-ellipsis);
    \draw[-To] (l-ellipsis) to[bend left=\bendamount] (psystems.west);
    \draw[-To] (real.east) to[bend right=\bendamount] (r-ellipsis);
    \draw[-To] (r-ellipsis) to[bend right=\bendamount] (psystems.east);

    \node[left=14mm of r-ellipsis,anchor=center,yshift=11mm,align=center] {general\\frameworks};
    \node[left=14mm of r-ellipsis,anchor=center,yshift=-11mm] {modelling};
  \end{tikzpicture}
  \caption{A graphical summary of our methodological conclusion:
    P systems are a powerful tool for constructing formal frameworks
    for other formalisms.}
  \label{fig:graphical-conclusion}
\end{figure}

We therefore propose that P systems should be used as a tool for
setting up general frameworks for reasoning about other formalisms,
which are more popular in biological modelling.  We give an example of
such a general framework---Boolean P systems---which capture Boolean
networks and in particular provide a homogeneous language for
sequential controllability.  Indeed, sequential controllability of
Boolean networks implicitly supposes the presence of a master
dynamical system emitting the control inputs.  Our Boolean P system
framework makes this master system explicit, as well as its
interactions with the controlled Boolean network.

The immediate future research direction which we have already started
is actually showing how Boolean P systems facilitate proving some
properties of sequential controllability of Boolean networks.
Another challenge would be capturing and reasoning about the ConEvs
dynamics of the control sequence~\cite{PardoID21}.  Under ConEvs, the
control is only allowed to evolve in a stable state, meaning that the
master dynamical system is not unilaterally acting on the Boolean
network any more, but both of them are part of feedback loop.

The main conclusion of our work is methodological: we believe that the
intrinsic flexibility and richness of P systems makes them an
excellent tool for constructing formal frameworks for other models
of computing.


\bibliographystyle{eptcs}
\bibliography{NCMA2022PvsB}


\bigskip

\begin{small}

\subsection*{Appendix 1: A Quantitative Study}
\label{sec:sota}

\begin{figure}[b!]
  \centering
  \includegraphics[width=.7\textwidth]{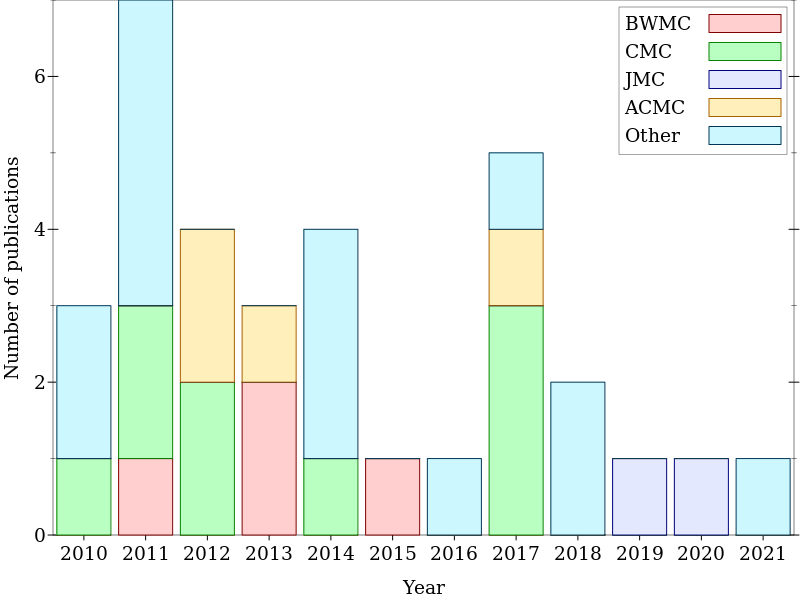}
  \caption{\small A breakdown by source of the 33 publications concerned with
    using P~systems to represent any kind of biological knowledge
    between years 2010 and 2021.  The bibliography behind the source
    ``Other'' is not exhaustive.}
  \label{fig:p-stats}
\end{figure}

To establish a comparative state of the art, we fixed the period
between years 2010 and 2021 and counted the publications using
P systems and Boolean networks for representing any kind of biological
knowledge.  Our choice of the time interval has a double motivation.
On the one hand, in 2010 P systems became a fully mature domain, and
the first international Conference on Membrane Computing was
organized.  On the other hand, Boolean networks started gaining
popularity in modelling and analysis over the same period of time.

For P systems, we focused mostly on the following sources,
representing the major bibliographical references of the domain:
\begin{itemize}
\item the bibliography of the Research Group on Natural
  Computing~\cite{RGNC},
\item the proceedings of the Brainstorming Weeks on Membrane Computing
  in Seville (BWMC), e.g.~\cite{BWMC2020},
\item the proceedings of the Conference on Membrane Computing (CMC),
  e.g.~\cite{CMC2020},
\item the Journal of Membrane Computing, e.g.~\cite{JMC2022},
\item the proceedings of the Asian Conference on Membrane Computing
  (ACMC), e.g.~\cite{ACMC2022}.
\end{itemize}
A quantitative synthesis of the relevant publications in these sources
is shown in Figure~\ref{fig:p-stats}.  This histogram indexes 33
publications.  The category ``Other'' refers to the papers which we
found cited in the indexed sources, and is not exhaustive.

For Boolean networks, we only focused on the publications in the
conference Computational Methods in Systems Biology,
e.g.~\cite{CMSB2021}, concerned with using Boolean networks to
represent any kind of biological knowledge.  We found 18 publications,
as shown in Figure~\ref{fig:cmsb-stats}.

\begin{figure}[h]
  \centering
  \includegraphics[width=.7\textwidth]{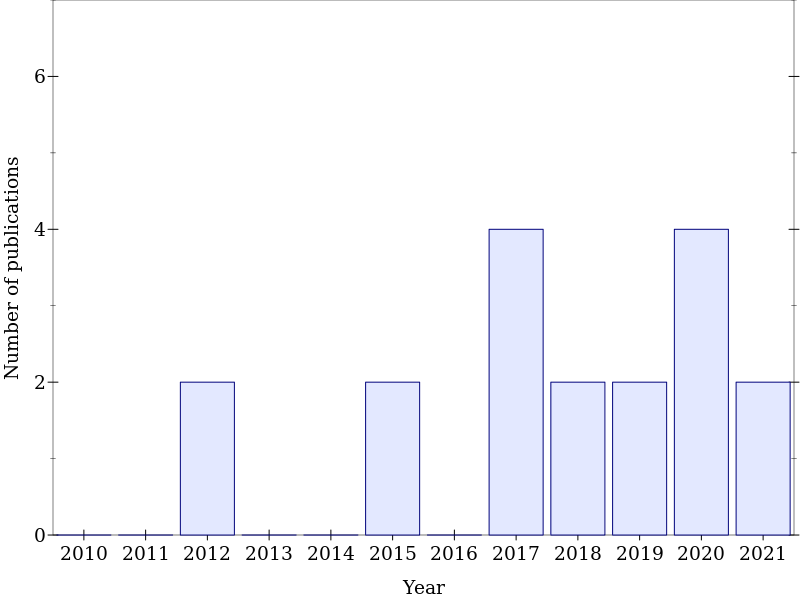}
  \caption{\small The distribution over the period 2010--2021 of the 18
    publications in the proceedings of the international conference
    Computational Methods in Systems Biology (CMSB) using Boolean
    networks to represent any kind of biological knowledge.}
  \label{fig:cmsb-stats}
\end{figure}

Full lists of indexed publications are given in the
following appendices.

The informal conclusion which we draw from this bibliographic study
comparing the number of publications in many major membrane computing
sources to the number of publications in a single systems biology
conference confirms the intuition from the introduction: Boolean
networks enjoy more success in biological modelling and analysis.

Even though explaining the deep reasons behind this disparity is
beyond the scope of our work, we speculate that the ultimate
simplicity of Boolean models and finiteness of the state space may
play a role.  Furthermore, the interest in Boolean modelling may be
traced back to the biological research (e.g.,~\cite{Thomas1973}), and
has developed in tight connection with biology
(e.g.,~\cite{Barabasi2011,Zanudo2018}).


\subsection*{Appendix 2}
\label{sec:app1-p-stats}

\newcommand{\pubtag}[1]{\tikz[baseline,yshift=1mm]\node[draw,semithick,rounded
  corners=1mm,inner sep=.5mm]{\small #1\strut};}

In this appendix, we list the 33 papers using P systems to represent
any kind of biological knowledge published between years 2010 and 2021
which were counted in Figure~\ref{fig:p-stats}.  The publications are
annotated by tags, representing the source:
\begin{itemize}[itemsep=.5mm]
\item \pubtag{rgnc}: the bibliography of the Research Group on Natural
  Computing,
\item \pubtag{bwmc}: the proceedings of the Brainstorming Weeks on
  Membrane Computing in Seville,
\item \pubtag{cmc}: the proceedings of the Conference on Membrane
  Computing,
\item \pubtag{jmc}: the Journal of Membrane Computing,
\item \pubtag{acmc}: the proceedings of the Asian Conference on
  Membrane Computing.
\end{itemize}

\subsubsection*{2021}
\label{sec:org64b1206}
\begin{enumerate}
\item García-Quismondo, M., Hintz W. D., Schuler M. S., \& Relyea
R. A. (2021):  Modeling Diel Vertical Migration with Membrane
Computing. Journal of Membrane Computing 3, 35--50.
\hfill \pubtag{rgnc}~\pubtag{cmc}
\end{enumerate}

\subsubsection*{2020}
\label{sec:orgbe3ca29}
\begin{enumerate}
\item Barbuti, R., Gori, R., Milazzo, P. et al. (2020): A survey of gene
regulatory networks modelling methods: from differential
equations, to Boolean and qualitative bioinspired
models. Journal of Membrane Computing 2, 207--226.  
\newline
\url{https://doi.org/10.1007/s41965-020-00046-y}
\hfill \pubtag{jmc}
\end{enumerate}

\subsubsection*{2019}
\label{sec:org1ed5f93}
\begin{enumerate}
\item Nash, A., Kalvala, S. (2019): A P system model of swarming and
aggregation in a Myxobacterial colony. Journal of Membrane Computing 1, 
103--111.  
\newline
\url{https://doi.org/10.1007/s41965-019-00015-0}
\hfill \pubtag{jmc}
\end{enumerate}

\subsubsection*{2018}
\label{sec:org24ad420}
\begin{enumerate}
\item Valencia-Cabrera, L., Graciani C., Pérez-Hurtado I.,
Pérez-Jiménez M. J., \& Riscos-Núñez A.  (2018):  A Decade of
Ecological Membrane Computing Applications. Bulletin of the
International Membrane Computing Society. 6, 39--50.
\hfill \pubtag{rgnc}

\item García-Quismondo, M., Graciani C., \& Riscos-Núñez A. (2018):
Membrane Computing as a Modelling Tool: Looking Back and
Forward from Sevilla. In: Carmen Graciani, Agustín Riscos-Núñez,
Gheorghe Păun, Gregorz Rozenberg, Arto Salomaa, editors: Enjoying
Natural Computing: Essays Dedicated to Mario de Jesús
Pérez-Jiménez on the Occasion of His 70th Birthday. 114--129.
\hfill \pubtag{rgnc}
\end{enumerate}

\subsubsection*{2017}
\label{sec:org04cd99e}
\begin{enumerate}
\item Cavaliere M., Sanchez A. (2017): The Evolutionary Resilience of
Distributed Cellular Computing. In: Leporati A., Rozenberg G.,
Salomaa A., Zandron C., editors: Membrane Computing. CMC 2016. Lecture
Notes in Computer Science, vol. 10105. Springer, Cham.  
\newline
\url{https://doi.org/10.1007/978-3-319-54072-6\_1}
\hfill \pubtag{cmc}

\item Hinze T.: Coping with Dynamical Structures for
Interdisciplinary Applications of Membrane Computing (2017). In:
Leporati A., Rozenberg G., Salomaa A., Zandron C., editors: Membrane
Computing. CMC 2016. Lecture Notes in Computer Science,
vol. 10105. Springer, Cham.  
\newline
\url{https://doi.org/10.1007/978-3-319-54072-6\_2}
\hfill \pubtag{cmc}

\item Barbuti R., Bove P., Milazzo P., Pardini G. (2017):  Applications
of P Systems in Population Biology and Ecology: The Cases of MPP
and APP Systems. In: Leporati A., Rozenberg G., Salomaa A.,
Zandron C., editors: Membrane Computing. CMC 2016. Lecture Notes in
Computer Science, vol. 10105. Springer, Cham.  
\newline
\url{https://doi.org/10.1007/978-3-319-54072-6\_3}
\hfill \pubtag{cmc}

\item Zhang G., Pérez-Jiménez M.J., Gheorghe M. (2017): Data Modeling
with Membrane Systems: Applications to Real Ecosystems. In:
Real-life Applications with Membrane Computing. Emergence,
Complexity and Computation, vol. 25. Springer, Cham.  
\newline
\url{https://doi.org/10.1007/978-3-319-55989-6\_7}

\item Mario J. Pérez-Jiménez (2017): Modelling the dynamics of complex
systems: A membrane computing based framework, Proceedings of 
the 6th Asian Conference on Membrane Computing, 2017.
\hfill \pubtag{acmc}
\end{enumerate}

\subsubsection*{2016}
\label{sec:org686a3c7}
\begin{enumerate}
\item Cristian Fondevilla, M. Àngels Colomer, Federico Fillat, Ulrike
Tappeiner (2016): Using a new PDP modelling approach for land-use and
land-cover change predictions: A case study in the Stubai Valley
(Central Alps), Ecological Modelling, vol. 322, pp.
101--114, ISSN 0304-3800, 
\newline
\url{https://doi.org/10.1016/j.ecolmodel.2015.11.016}.
\end{enumerate}

\subsubsection*{2015}
\label{sec:org142af88}
\begin{enumerate}
\item Gheorghe Păun (2011): Looking for Computer in the Biological Cell. After
Twenty Years, Proceedings of the Ninth Brainstorming Week on
Membrane Computing, 251--300.
\hfill \pubtag{bwmc}
\end{enumerate}

\subsubsection*{2014}
\label{sec:org0aaea39}
\begin{enumerate}
\item Colomer, A. M., Margalida A., Valencia-Cabrera L., \& Palau
A. (2014):  Application of a computational model for complex
fluvial ecosystems: The population dynamics of zebra mussel
Dreissena polymorpha as a case study. Ecological
Complexity 20, 116--126.
\hfill \pubtag{rgnc}

\item Frisco, P., Gheorghe M., \& Pérez-Jiménez M. J. (2014):
Applications of Membrane Computing in Systems and Synthetic
Biology. Emergence, Complexity and Computation. 7, 266.
\hfill \pubtag{rgnc}

\item Blakes, J., Twycross J., Konur S., Romero-Campero F. J.,
Krasnogor N., \& Gheorghe M.(2014):  Infobiotics Workbench:
A P Systems Based Tool for Systems and Synthetic
Biology. Applications of Membrane Computing in Systems and
Synthetic Biology 7, 1--42 .
\hfill\pubtag{rgnc}

\item Pérez-Jiménez M.J. (2014): A Bioinspired Computing Approach to
Model Complex Systems. In: Gheorghe M., Rozenberg G., Salomaa
A., Sosík P., Zandron C., editors: Membrane
Computing. CMC 2014. Lecture Notes in Computer Science,
vol. 8961. Springer, Cham.  
\newline
\url{https://doi.org/10.1007/978-3-319-14370-5\_2}
\hfill \pubtag{cmc}
\end{enumerate}

\subsubsection*{2013}
\label{sec:org018864d}
\begin{enumerate}
\item Ardelean, I., Díaz-Pernil D., Gutiérrez-Naranjo M. A.,
Peña-Cantillana F., \& Sarchizian I. (2013):  Studying the
Chlorophyll Fluorescence in Cyanobacteria with Membrane
Computing Techniques. Eleventh Brainstorming Week on Membrane
Computing (11BWMC), 9--24.
\hfill \pubtag{rgnc}~\pubtag{bwmc}

\item L. Valencia-Cabrera, M. García-Quismondo, M.J. Pérez-Jiménez,
Y. Su, H. Yu, L. Pan  (2011): Analysing Gene Networks with PDP
Systems. Arabidopsis thaliana, a Case Study. Proceedings of the
Ninth Brainstorming Week on Membrane Computing, 257--272.
\hfill \pubtag{bwmc}

\item Colomer M.À., Margalida A., Pérez-Jiménez M.J.  (2013): 
Population Dynamics
P system (PDP) models: a standardized protocol for describing
and applying novel bio-inspired computing tools. Plos
one. 8(4):e60698. 
\newline
DOI: 10.1371/journal.pone.0060698. PMID:
23593284; PMCID: PMC3622025.
\hfill \pubtag{acmc}
\end{enumerate}

\subsubsection*{2012}
\label{sec:org03eeb3a}
\begin{enumerate}
\item García-Quismondo, M., Valencia-Cabrera L., Su Y., Pérez-Jiménez
M. J., Pan L., \& Yu H.  (2012):  Modeling logic gene networks by
means of probabilistic dynamic P systems. In: Linqiang Pan,
Gheorghe Paun, Tao Song, editors. Asian Conference on Membrane
Computing. 30--60 (2012).
\hfill \pubtag{rgnc} \pubtag{acmc}

\item Romero-Campero, F. J., \& Pérez-Jiménez M.J. (2012):  P systems
as a modeling framework for molecular Systems
Biology. In: Linqiang Pan, Gheorghe Paun, Tao Song, editors: Asian
Conference on Membrane Computing. 8--10.
\hfill \pubtag{rgnc} \pubtag{acmc}

\item Martínez-del-Amor M.A. et al. (2013): DCBA: Simulating Population
Dynamics P Systems with Proportional Object Distribution. In:
Csuhaj-Varjú E., Gheorghe M., Rozenberg G., Salomaa A., Vaszil
G., editors: Membrane Computing. CMC 2012. Lecture Notes in Computer
Science, vol. 7762. Springer, Berlin, Heidelberg.  
\newline
\url{https://doi.org/10.1007/978-3-642-36751-9\_18}
\hfill \pubtag{rgnc} \pubtag{cmc}

\item Ramón P., Troina A. (2013): Modelling Ecological Systems with the
Calculus of Wrapped Compartments. In: Csuhaj-Varjú E., Gheorghe
M., Rozenberg G., Salomaa A., Vaszil G., editors: Membrane
Computing. CMC 2012. Lecture Notes in Computer Science,
vol. 7762. Springer, Berlin, Heidelberg.  
\newline
\url{https://doi.org/10.1007/978-3-642-36751-9\_24}
\hfill \pubtag{cmc}
\end{enumerate}

\subsubsection*{2011}
\label{sec:orgc3392e7}
\begin{enumerate}
\item Colomer, A. M., Lavín S., Marco I., Margalida A.,
Pérez-Hurtado I., Pérez-Jiménez M. J., et al. (2011):
Modeling population growth of Pyrenean Chamois (Rupicapra
p. pyrenaica) by using P systems. Lecture Notes in Computer
Science. 6501, 144--159.
\hfill \pubtag{rgnc}

\item Gheorghe, M., Manca V., \& Romero-Campero F. J. (2011):
Deterministic and stochastic P systems for modelling cellular
processes. Natural Computing. 9(2), 457--473.
\hfill \pubtag{rgnc}

\item Colomer, A. M., Pérez-Hurtado I., Riscos-Núñez A., \&
Pérez-Jiménez M. J. (2011):  Comparing simulation algorithms
for multienvironment probabilistic P system over a standard
virtual ecosystem. Natural Computing 11, 369--379.
\ \hfill \pubtag{rgnc}

\item Cardona, M., Colomer M. A., Margalida A., Palau A.,
Pérez-Hurtado I., Pérez-Jiménez M. J., et al. (2011):
A computational modeling for real ecosystems based on
P systems. Natural Computing 10(1), 39--53.

\hfill \pubtag{rgnc}

\item M.A. Colomer, C. Fondevilla, L. Valencia-Cabrera (2011): 
A New P System to Model the Subalpine and Alpine Plant Communities, 
Proceedings of the Ninth Brainstorming Week on Membrane Computing,
91--112.
\hfill \pubtag{bwmc}

\item Beal J.:  Bridging Biology and Engineering Together with
Spatial Computing (2012). In: Gheorghe M., Păun Gh., Rozenberg G.,
Salomaa A., Verlan S., editors Membrane
Computing. CMC 2011. Lecture Notes in Computer Science,
vol. 7184. Springer, Berlin, Heidelberg.  
\newline
\url{https://doi.org/10.1007/978-3-642-28024-5\_2}
\hfill \pubtag{cmc}

\item Giavitto J.L. (2012): The Modeling and the Simulation of the Fluid
Machines of Synthetic Biology. In: Gheorghe M., Păun Gh.,
Rozenberg G., Salomaa A., Verlan S., editors: Membrane
Computing. CMC 2011. Lecture Notes in Computer Science,
vol. 7184. Springer, Berlin, Heidelberg.  
\newline
\url{https://doi.org/10.1007/978-3-642-28024-5\_3}
\hfill \pubtag{cmc}
\end{enumerate}

\subsubsection*{2010}
\label{sec:org6c28d16}
\begin{enumerate}
\item Colomer, A. M., Lavín S., Marco I., Margalida A.,
Pérez-Hurtado I., Pérez-Jiménez M. J., et al. (2010):
Modeling population growth of Pyrenean Chamois (Rupicapra
p. pyrenayca) by using P systems. In: Marian Gheorghe, Thomas
Hinze, Gheorghe Păun, editors: Eleventh International Conference
on Membrane Computing (CMC11). 121--135. 
\hfill \pubtag{rgnc} \pubtag{cmc}

\item Cardona, M., Colomer A. M., Margalida A., Pérez-Hurtado I.,
Pérez-Jiménez M. J., \& Sanuy D. (2010):  A P system based
model of an ecosystem of some scavenger birds. Lecture Notes
in Computer Science, vol. 5957, 182--195.
\hfill \pubtag{rgnc}

\item Besozzi D., Cazzaniga P., Mauri G., Pescini D. (2010):
BioSimWare: A Software for the Modeling, Simulation and Analysis
of Biological Systems. In: Gheorghe M., Hinze T., Păun Gh.,
Rozenberg G., Salomaa A., editors: Membrane
Computing. CMC 2010. Lecture Notes in Computer Science,
vol. 6501. Springer, Berlin, Heidelberg.  
\newline
\url{https://doi.org/10.1007/978-3-642-18123-8\_12}
\hfill \pubtag{cmc}
\end{enumerate}


\subsection*{Appendix 3}
\label{sec:cmsb-stats}

In this appendix, we list the 18 papers using Boolean networks to
represent any kind of biological knowledge, published between the
years 2010 and 2021 in the proceedings of the international conference
on Computational Methods in Systems Biology (CMSB), and which were
counted in Figure~\ref{fig:cmsb-stats}.

\subsubsection*{2021}
\label{sec:org800ff97}
\begin{enumerate}
\item Biswas A., Gupta A., Missula M., Thattai M. (2021): Automated
Inference of Production Rules for Glycans. In: Cinquemani E.,
Paulevé L., editors: Computational Methods in Systems
Biology. CMSB 2021. Lecture Notes in Computer Science,
vol. 12881. Springer, Cham.  
\newline
\url{https://doi.org/10.1007/978-3-030-85633-5\_4}

\item Thuillier K., Baroukh C., Bockmayr A., Cottret L., Paulevé L.,
Siegel A. (2021): Learning Boolean Controls in Regulated
Metabolic Networks: A Case-Study. In: Cinquemani E., Paulevé
L., editors: Computational Methods in Systems
Biology. CMSB 2021. Lecture Notes in Computer Science,
vol. 12881. Springer, Cham.  
\newline
\url{https://doi.org/10.1007/978-3-030-85633-5\_10}
\end{enumerate}

\subsubsection*{2020}
\label{sec:org8021518}
\begin{enumerate}
\item Cifuentes Fontanals L., Tonello E., Siebert H. (2020): Control
Strategy Identification via Trap Spaces in Boolean Networks. In:
Abate A., Petrov T., Wolf V., editors: Computational Methods in
Systems Biology. CMSB 2020. Lecture Notes in Computer Science,
vol. 12314. Springer, Cham. 
\newline
\url{https://doi.org/10.1007/978-3-030-60327-4\_9}

\item Diop O., Chaves M., Tournier L. (2020): Qualitative Analysis of
Mammalian Circadian Oscillations: Cycle Dynamics and
Robustness. In: Abate A., Petrov T., Wolf V., editors: Computational
Methods in Systems Biology. CMSB 2020. Lecture Notes in Computer
Science, vol. 12314. Springer, Cham.  
\newline
\url{https://doi.org/10.1007/978-3-030-60327-4\_10}

\item Chevalier S., Noël V., Calzone L., Zinovyev A., Paulevé
L. (2020): Synthesis and Simulation of Ensembles of Boolean
Networks for Cell Fate Decision. In: Abate A., Petrov T., Wolf
V., editors: Computational Methods in Systems
Biology. CMSB 2020. Lecture Notes in Computer Science,
vol. 12314. Springer, Cham.  
\newline
\url{https://doi.org/10.1007/978-3-030-60327-4\_11}

\item Su C., Pang J. (2020): Sequential Temporary and Permanent Control
of Boolean Networks. In: Abate A., Petrov T., Wolf V., editors:
Computational Methods in Systems Biology. CMSB 2020. Lecture
Notes in Computer Science, vol. 12314. Springer, Cham.  
\newline
\url{https://doi.org/10.1007/978-3-030-60327-4\_13}
\end{enumerate}

\subsubsection*{2019}
\label{sec:orge907130}
\begin{enumerate}
\item Mandon H., Su C., Haar S., Pang J., Paulevé L. (2019): Sequential
Reprogramming of Boolean Networks Made Practical. In: Bortolussi
L., Sanguinetti G., editors: Computational Methods in Systems
Biology. CMSB 2019. Lecture Notes in Computer Science,
vol. 11773. Springer, Cham. 
\newline
 \url{https://doi.org/10.1007/978-3-030-31304-3\_1}
 
\item Pardo J., Ivanov S., Delaplace F. (2019): Sequential
Reprogramming of Biological Network Fate. In: Bortolussi L.,
Sanguinetti G., editors: Computational Methods in Systems
Biology. CMSB 2019. Lecture Notes in Computer Science,
vol. 11773. Springer, Cham.  
\newline
\url{https://doi.org/10.1007/978-3-030-31304-3\_2}
\end{enumerate}

\subsubsection*{2018}
\label{sec:org5e5ba4c}
\begin{enumerate}
\item Razzaq M., Kaminski R., Romero J., Schaub T., Bourdon J.,
Guziolowski C. (2018): Computing Diverse Boolean Networks from
Phosphoproteomic Time Series Data. In: Češka M., Šafránek
D., editors: Computational Methods in Systems
Biology. CMSB 2018. Lecture Notes in Computer Science,
vol. 11095. Springer, Cham.  
\newline
\url{https://doi.org/10.1007/978-3-319-99429-1\_4}

\item Paul S., Pang J., Su C. (2018): On the Full Control of Boolean
Networks. In: Češka M., Šafránek D., editors: Computational Methods
in Systems Biology. CMSB 2018. Lecture Notes in Computer
Science, vol. 11095. Springer, Cham.  
\newline
\url{https://doi.org/10.1007/978-3-319-99429-1\_21}
\end{enumerate}

\subsubsection*{2017}
\label{sec:org9b6446b}
\begin{enumerate}
\item Biane C., Delaplace F. (2017): Abduction Based Drug Target
Discovery Using Boolean Control Network. In: Feret J., Koeppl
H., editors: Computational Methods in Systems
Biology. CMSB 2017. Lecture Notes in Computer Science,
vol. 10545. Springer, Cham.  
\newline
\url{https://doi.org/10.1007/978-3-319-67471-1\_4}

\item Carcano A., Fages F., Soliman S. (2017): Probably Approximately
Correct Learning of Regulatory Networks from Time-Series
Data. In: Feret J., Koeppl H., editors: Computational Methods in
Systems Biology. CMSB 2017. Lecture Notes in Computer Science,
vol. 10545. Springer, Cham.  
\newline
\url{https://doi.org/10.1007/978-3-319-67471-1\_5}

\item Mandon H., Haar S., Paulevé L. (2017): Temporal Reprogramming of
Boolean Networks. In: Feret J., Koeppl H., editors: Computational
Methods in Systems Biology. CMSB 2017. Lecture Notes in Computer
Science, vol. 10545. Springer, Cham.  
\newline
\url{https://doi.org/10.1007/978-3-319-67471-1\_11}

\item Paulevé L. (2017): Pint: A Static Analyzer for Transient Dynamics
of Qualitative Networks with IPython Interface. In: Feret J.,
Koeppl H., editors: Computational Methods in Systems
Biology. CMSB 2017. Lecture Notes in Computer Science,
vol. 10545. Springer, Cham.  
\newline
\url{https://doi.org/10.1007/978-3-319-67471-1\_20}
\end{enumerate}

\subsubsection*{2015}
\label{sec:orgf35d300}
\begin{enumerate}
\item Ostrowski M., Paulevé L., Schaub T., Siegel A., Guziolowski
C. (2015): Boolean Network Identification from Multiplex Time
Series Data. In: Roux O., Bourdon J., editors: Computational Methods
in Systems Biology. CMSB 2015. Lecture Notes in Computer
Science, vol. 9308. Springer, Cham.  
\newline
\url{https://doi.org/10.1007/978-3-319-23401-4\_15}

\item Abou-Jaoudé W., Feret J., Thieffry D. (2015): Derivation of
Qualitative Dynamical Models from Biochemical Networks. In: Roux
O., Bourdon J., editors: Computational Methods in Systems
Biology. CMSB 2015. Lecture Notes in Computer Science,
vol. 9308. Springer, Cham.  
\newline
\url{https://doi.org/10.1007/978-3-319-23401-4\_17}
\end{enumerate}

\subsubsection*{2012}
\label{sec:org05493e9}
\begin{enumerate}
\item Folschette M., Paulevé L., Inoue K., Magnin M., Roux O. (2012): 
Concretizing the Process Hitting into Biological Regulatory
Networks. In: Gilbert D., Heiner M. (eds) Computational Methods
in Systems Biology. CMSB 2012. Lecture Notes in Computer
Science, vol. 7605. Springer, Berlin.
\newline
Heidelberg. \url{https://doi.org/10.1007/978-3-642-33636-2\_11}

\item Naldi A., Monteiro P.T., Chaouiya C. (2012): Efficient Handling
of Large Signalling-Regulatory Networks by Focusing on Their
Core Control. In: Gilbert D., Heiner M., editors: Computational
Methods in Systems Biology. CMSB 2012. Lecture Notes in Computer
Science, vol. 7605. Springer, Berlin, Heidelberg . 
\newline
\url{https://doi.org/10.1007/978-3-642-33636-2\_17}
\end{enumerate}
\end{small}
\end{document}